\documentclass[letterpaper, 10pt, conference]{ieeeconf}
\IEEEoverridecommandlockouts
\usepackage{xargs}                      % Use more than one optional parameter in a new commands
\usepackage[pdftex,dvipsnames]{xcolor}  % Coloured text etc.
\usepackage[utf8]{inputenc}
\usepackage[english]{babel}
\usepackage{amsmath,amssymb,hyperref,xcolor}

\usepackage{amsthm}
\usepackage{enumerate}
\usepackage{subcaption}
\usepackage{refcount}
\usepackage[noadjust]{cite}
\usepackage{cleveref}
\crefrangeformat{equation}{#3(#1)#4--#5(#2)#6}

\usepackage[colorinlistoftodos,prependcaption,textsize=normalsize]{todonotes}
\newcommandx{\unsure}[2][1=]{\todo[linecolor=red,backgroundcolor=red!25,bordercolor=red,#1]{#2}}
\newcommandx{\change}[2][1=]{\todo[linecolor=blue,backgroundcolor=blue!25,bordercolor=blue,#1]{#2}}
\newcommandx{\info}[2][1=]{\todo[linecolor=OliveGreen,backgroundcolor=OliveGreen!25,bordercolor=OliveGreen,#1]{#2}}
\newcommandx{\improvement}[2][1=]{\todo[linecolor=Plum,backgroundcolor=Plum!25,bordercolor=Plum,#1]{#2}}
\newcommandx{\thiswillnotshow}[2][1=]{\todo[disable,#1]{#2}}
\newcommand{\removed}[1]{{}}

\usepackage{stfloats}

\usepackage{tikz}
\usepackage{pgfplots}
\usepackage{pgfgantt}
\usepackage{pdflscape}
\pgfplotsset{compat=newest}
\pgfplotsset{plot coordinates/math parser=false}
\title{\LARGE \bf {Filtering in Projection-based Integrators \\for Improved Phase Characteristics}}

\author{Hoang Chu, S.J.A.M van den Eijnden, M.F. Heertjes, W.P.M.H. Heemels   %
\thanks{The authors are with the Control Systems Technology Section, Dept. Mechanical Engineering, Eindhoven University of Technology, The Netherlands. M.F. Heertjes is also with ASML,  Mechatronics \& Measurements Systems, Veldhoven, The Netherlands. Corresponding author: Hoang Chu
({\tt\small h.chu@tue.nl})
\newline \indent This research received funding form the European Research Council (ERC) under the Advanced ERC grant agreement PROACTHIS, no. 101055384.}%
}

\newtheorem{theorem}{Theorem}

\newtheorem{remark}{Remark}

\usepackage{accents}

\let\leq\leqslant
\let\geq\geqslant

\let\cal\mathcal

\newcommand{\remove}[1]{ }
\newcommand{\ree}{\mathbb{R}}

\newcommand{\cS}{{\cal S}}

\newcommand{\xh}{{x_h}}

\newcommand{\dxh}{\delta{x_h}}
\newcommand{\dv}{\delta{v}}

\DeclareMathOperator*{\argmin}{\arg\!\min}

\begin{document}
\maketitle
\thispagestyle{empty}
\pagestyle{empty}

\begin{abstract}
Projection-based integrators are effectively employed in high-precision systems with growing industrial success. By utilizing a projection operator, the resulting projection-based integrator keeps its input-output pair within a designated sector set, leading to unique freedom in control design that can be directly translated into performance benefits. This paper aims to enhance projection-based integrators by incorporating well-crafted linear filters into its structure, resulting in a new class of projected integrators that includes the earlier ones, such as the hybrid-integrator gain systems (with and without pre-filtering) as special cases. The extra design freedom in the form of two filters in the input paths to the projection operator and the internal dynamics allows the controller to break away from the inherent limitations of the linear control design. The enhanced performance properties of the proposed structure are formally demonstrated through a (quasi-linear) describing function analysis, the absence of the gain-loss problem, and numerical case studies showcasing improved time-domain properties. The describing function analysis is supported by rigorously showing incremental properties of the new filtered projection-based integrators thereby guaranteeing that the computed steady-state responses are unique and asymptotically stable. 
% \textcolor{red}{something about system theoretic fundamental properties here? Stabiity..}
\end{abstract}

\tikzstyle{block} = [draw, fill=blue!20, rectangle, 
    minimum height=3em, minimum width=6em]
\tikzstyle{sum} = [draw, fill=blue!20, circle, node distance=1cm]
\tikzstyle{input} = [coordinate]
\tikzstyle{output} = [coordinate]
\tikzstyle{pinstyle} = [pin edge={to-,thin,black}]

\section{Introduction}
% set up context
To meet the ever-increasing performance demands of high-precision motion systems, the use of control systems is key. Nowadays, linear controllers such as proportional-integral-derivative (PID) controllers are still the standard in many industries. The reason for the widespread use of linear controllers seems attributed to the transparency in analysis and design, and their ease of use. Linear controllers, however, are subject to fundamental limitations such as Bode's gain-phase realtionship and other frequency- and time-domain limitations \cite{SeronLimit}. Nonlinear controllers, on the other hand, are not limited in the same way as LTI controllers, suggesting possibilities to overcome the classical limitations of linear control, and thereby to realize unparralelled control performance (\cite{Heertjes_2020}). Of particular interest are nonlinear integrators that preserve the interpretation and functionality of a linear integrator with enhanced phase properties due to the nonlinear nature. Examples include integrators with a state reset map (e.g., the Clegg integrator \cite{Clegg,NESIC2008}), integrators with discontinuous input (e.g., \cite{zeroPhaseShift}), split-path integrators \cite{SPANoriginal,ShaMaa_TCST22a}, and projection-based integrators (PBIs) \cite{DeeSha_AUT21a,Shi_2022}.

% talk about HIGS
Projection-based integrators have been gaining attention recently due to their favorable properties from a system theoretical perspective and their potential in industrial applications. 
PBIs are employed in controllers with input constraints (and built-in anti-windup mechanisms), see, e.g., \cite{LORENZETTI2022110606,ZaoPBC}, or in controllers enforcing input-output properties for direct performance enhancements \cite{DeeSha_AUT21a}. In this paper, we focus on the latter type of PBI, because a notable example of this type - hybrid integrator-gain systems (HIGS) \cite{DeeSha_AUT21a} - has proven successful in industrial applications such as wafer scanners \cite{Heertjes_2020}, micro-electro-mechanical systems \cite{Shi_2022,shi2023negative}, and wire-bonders \cite{EijkBeerHODFHIGS}. The key idea underlying this projection-based integrator is to keep its input-output pair in a specified sector set by means of projection.
As a result of applying such a sector-based projection, the input and output of HIGS always have equal sign. In turn, sign equivalence yields a significant reduction in phase lag from $90$ degrees in LTI integrators to $38.1$ degrees in projection-based integrators as understood from a describing function perspective \cite{EijCCTA}. This phase lag reduction allows for additional freedom in controller design since reduced phase lag allows for increased controller gains, resulting in improvements in  system bandwidth and disturbance rejection properties. Although the gain and the phase of projection-based integrators no longer adhere to the classical constraints due to the Bode gain-phase relationship (i.e., a 20dB-per-decade magnitude slope paired with $90$-degree phase lag) and increased performance compared to LTI control, there is a strong interest in creating further flexibility in their design.

\begin{figure}[tb!]
    \centering
    \includegraphics[width = 0.45\textwidth]{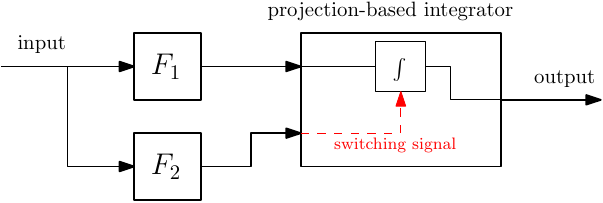}
    \caption{Projection-based integrator. The LTI filter $F_1$ adheres to the ``gain'' characteristics of the element, whereas $F_2$ determines the ``phase'' characteristics. The switching signal, affected by the choice for $F_2$, determines when projection is applied to the integrator state.}
    \label{fig:filter HIGS structure}
\end{figure}

To provide additional freedom in design, we propose a novel structure of projection-based integrators. The general idea is depicted in Fig.~\ref{fig:filter HIGS structure},  where the LTI filters $F_1$ and $F_2$ are key for gain and phase tuning, respectively. This setup might remind of a structure used in \cite{Nima} in the context of reset control or in the context of filter split-path nonlinear integrators \cite{ShaMaa_TCST22a}. Interestingly, projection-based integrators result in continuous control outputs (which are more amenable to analysis), rather than discontinuous ones as obtained with the aforementioned strategies. To elaborate, the filter $F_2$ is used to regulate the switching behavior (instances of projection) of the projection element. That is, rather than switching based on the input directly (as in original schemes \cite{DeeSha_AUT21a}), the novelty of the element is in switching based on the signal filtered by $F_2$. Switching in PBIs is typically associated with phase properties, and thus manipulating the switching signal through the fiter $F_2$ allows for tuning the phase. The filter $F_1$ mainly affects how the input to the PBI is entering the internal integrator dynamics. Noteworthy is that this structure has the standard HIGS \cite{DeeSha_AUT21a}  and the cascade of a prefilter $F$ and HIGS as used in \cite{EijHee_CSL20a} to overcome the fundamental performance limitations of LTI controllers, as special cases by taking  $F_1=F_2=1$ or taking $F_1=F_2=F$. Here we decouple the two filters, allowing further performance benefits. %, as we will see. 

In line with the above, the main contributions of this paper are summarized as follows. \begin{enumerate}[(i)]
\item First, we introduce the novel projection-based integrator, termed Filtered HIGS (abbreviated as FHIGS) due to its connections with HIGS and provide the appropriate mathematical model of FHIGS within the framework of extended projected dynamical systems (ePDS) \cite{DeeSha_AUT21a}. 
\item We derive an equivalent piecewise linear (PWL) model of FHIGS useful for further analysis, 
\item We calculate the describing function (DF) of FHIGS using the PWL model, and 
\item We provide incremental stability guarantees of FHIGS leading to the existence of a unique steady-state response to a sinusoidal input, thereby providing a strong fundamental base for the DF. 
\item Finally, we demonstrate the enhanced performance properties of the proposed FHIGS structure with numerical case studies through a (quasi-linear) DF comparison to HIGS and LTI controllers, the mitigation of the gain-loss problem \cite{HeeEij_CCTA21a}, and improved time-domain properties. %{\bf check if this is correct}
\end{enumerate}

% , we aim to generalize HIGS by separating the gain and phase of the nonlinear structure using different filters. This new structure (see Fig.~\ref{fig:filter HIGS structure}) - named filtered hybrid integrator-gain system (hereafter referred to as Filtered HIGS, or FHIGS) - uses an additional filter ($F_2$) to regulate the switching behavior of the nonlinear integrator. The FHIGS element, therefore, inherits the gain from $F_1$ and the phase from $F_2$, thus allowing for (almost) separate tuning of gain and phase, similar to the benefit of the shaping filter in \cite{ShapingFilterResetHosseinNia}. With this separate gain-phase structure, FHIGS can have more phase advantage than HIGS by switching to the gain mode earlier, while having approximately the same gain. FHIGS can also avoid the gain-loss problem of HIGS \cite{HeeEij_CCTA21a}, due to the capability of suppressing certain frequencies using $F_2$. 

% \begin{figure}[bp]
%     \centering
%     \includegraphics[width = 0.48\textwidth]{Figs/filteredHIGS_2branch.pdf}
%     \caption{A control system with Filtered HIGS. The systems $F_1,F_2,F_3,G$ are LTI. The red branch determines the "gain" of FHIGS, while the blue branch determines the "phase" of FHIGS.}
%     \label{fig:filter HIGS structure}
% \end{figure}

% paper structure
The remainder of the paper is as follows. We introduce FHIGS in Section \ref{sec:ePDS representation} and formalize this in the ePDS framework. A more explicit model in the form of a PWL system is provided in Section \ref{sec: PWL openloop}. In addition to these modeling contributions, we will explicitly compute a steady-state response of the FHIGS to a sinusoidal input, which will be used to compute its describing function in Section \ref{sec: DF}. To justify the describing function, we will show incremental properties of FHIGS in Section \ref{sec: convergence}. We proceed to show some numerical examples to illustrate the benefits of filtered projection integrators in Section \ref{sec: examples} and conclude the paper in Section \ref{sec: conclusion}.

\section{Projected dynamics of FHIGS} \label{sec:ePDS representation}
In this section, we formally introduce FHIGS. We briefly describe the control system, show how a projection operator benefits the response of a nonlinear integrator in the time domain, and proceed to give a representation of a control system with FHIGS in terms of an extended projected dynamical system \cite{heemelsaneel_lcss_2023}.

\subsection{System description}
\begin{figure}[tbp]
    \centering
    \smallskip\smallskip
    \includegraphics[width = 0.48\textwidth]{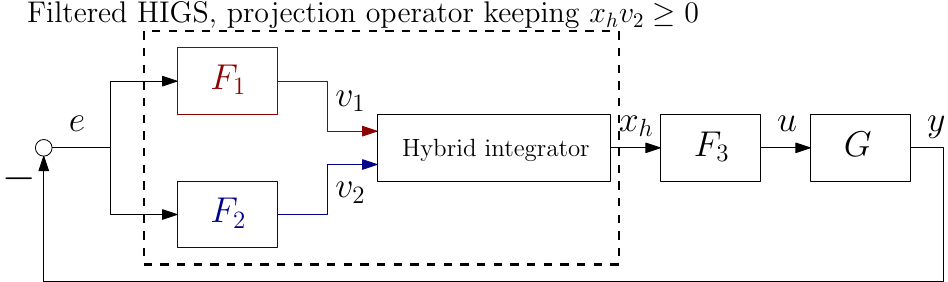}
    \caption{Filtered HIGS in closed-loop with other linear systems $F_3$ and $G$}
    \label{fig:FHIGS closedloop}
\end{figure}
We consider the closed-loop system depicted in Fig.~\ref{fig:FHIGS closedloop} where $F_1$, $F_2$, $F_3$, and $G$ are SISO LTI systems. The cascade of the loop-shaping filter $F_3$ and the linear system $G$ is considered as the plant for which the dynamics are given by
\begin{equation}
\begin{aligned}
    \dot x_p &= A_p x_p + B_p x_h,\\
   y &= C_p x_p,
\end{aligned}
\end{equation}
with $x_p \in \ree^{n_p}$ the states, $x_h \in \mathbb{R}$ the (control) input, and $y \in \mathbb{R}$ the output. %, and $C_p B_p = 0$. 
The starting point of FHIGS is a  first-order dynamics (before introducing the projection operator), combined with state-space models for the filters $F_1$ and $F_2$, resulting in the linear unprojected dynamics
\begin{equation} \label{eq: open unprojected fhigs}
\begin{aligned}
    \begin{bmatrix}
        \dot x_h \\ \dot x_{v_1} \\ \dot x_{v_2} 
    \end{bmatrix} &= \begin{bmatrix}
        -\alpha_h &\omega_h C_{v_1} &0\\ 0 &A_{v_1} &0\\ 0&0 &A_{v_2}
    \end{bmatrix} \begin{bmatrix}
        x_h\\ x_{v_1} \\ x_{v_2}
    \end{bmatrix} + \begin{bmatrix}
        \omega_h D_{v_1} \\ B_{v_1}\\ B_{v_2}
    \end{bmatrix} e\\
    &=: f_c(x_c,e)\\
    v_1 &= C_{v_1} x_{v_1}+ D_{v_1} e,\\
    v_2 &= C_{v_2} x_{v_2}+ D_{v_2} e,
\end{aligned}
\end{equation}
with filter states $x_{v_1} \in \ree^{n_{v_1}}$ and  $x_{v_2} \in \ree^{n_{v_2}}$ for filter $F_1$ and $F_2$, respectively, first-order dynamics state $x_h \in \ree$, overall controller state $x_c = \begin{bmatrix}
    x_h^\top & x_{v_1}^\top & x_{v_2}^\top
\end{bmatrix}^\top$, and controller input $e\in \mathbb{R}$. Note that for $\alpha_h=0$, we obtain integrator dynamics. Here, $\omega_h >0$ is a controller parameter in the first-order dynamics. % {\bf we always talk about integrator dynamics, but we have first-order dynamics, both in intro, abstract and main text... why? }.
Furthermore, $(A_{v_i}, B_{v_i}, C_{v_i}, D_{v_i})$, $i=1,2$ are matrix quadruples of appropriate dimensions describing the filter dynamics. 

The \emph{unprojected} closed-loop dynamics with the interconnection $e = -y$ are given by
\begin{equation}\label{eq:CL}
    \begin{bmatrix}
        \dot x_c \\ \dot x_p 
    \end{bmatrix} = \begin{bmatrix}
        f_c(x_c, -C_p x_p)\\
        A_p x_p + B_px_h
    \end{bmatrix}=: f(x),
\end{equation}
with $x = \begin{bmatrix}
    x_c^\top &x_p^\top
\end{bmatrix}^\top.$ Next, we will motivate and introduce projection into the closed-loop control system \eqref{eq:CL}. 

\subsection{Motivation of using the projection operator}

\begin{figure}[tbp]
    \centering
    \smallskip\smallskip\includegraphics[width = 0.15\textwidth]{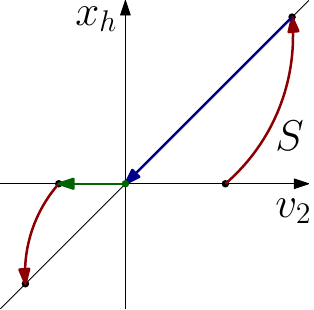}
    \caption{Example of FHIGS trajectory in a sector set}
    \label{fig:fhigs trajectory}
\end{figure}

\begin{figure}[tp]
    \centering
    \includegraphics[width=0.45\textwidth]{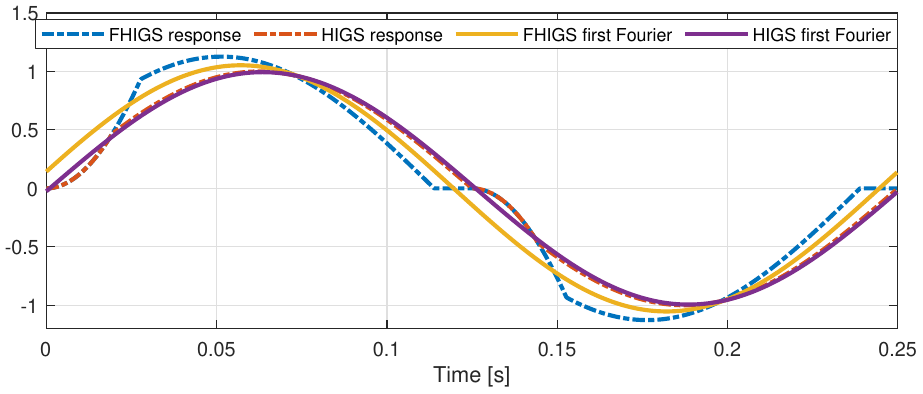}
    \caption{HIGS and FHIGS steady-state sinusoidal responses with their first Fourier approximations. FHIGS (yellow) has approximately the same gain and more phase lead compared with HIGS (purple)}
    \label{fig: Fourier}
\end{figure}

% details about FHIGS: sector, ePDS, PWL
The rationale underlying FHIGS is to keep the {\em $F_2$-filtered input} and the first-order-dynamics state in a sector set using a projection operator. The sector set is defined as
\begin{equation}
    S = \{(x_h,v_2)\in \ree^2 \mid (x_h-k_1 v_2)(x_h-k_2 v_2) \leq 0\},
\end{equation}
(see Fig.~\ref{fig:fhigs trajectory} for the set $S$ with $k_1=0,k_2=1$).
When the trajectory of $(x_h,v_2)$ is in the interior of the sector set, the first order dynamics are active ($\dot x_h =-\alpha_h x_h + \omega_h v_1$), and the projection operator does not have any effect. When the trajectory is at the boundary and tends to move outside the sector set, the dynamics are ``projected'' in such a way that the trajectory remains on the boundary of the sector set (see Fig.~\ref{fig:fhigs trajectory}) -- below we will formalize the projection operator. Since the sector set now depends on the filtered input $v_2$, it is possible to change the switching instants to achieve more phase advantage from a describing function perspective (see Fig.~\ref{fig: Fourier} which depicts the first harmonics of the output of FHIGS to a sinusoidal input) compared to classical HIGS, which corresponds to the case where the filters $F_1$ and $F_2$ are simply constant gains equal to $1$.

\subsection{ePDS formalization of a control system with FHIGS}

We formalize the dynamics of the control system with the proposed FHIGS element in the context of ePDS \cite{DeeSha_AUT21a}. This ePDS framework forms an extension to the classical projected dynamical systems (\cite{HENRY1972545,nagurney1995projected}), by allowing partial projection of the dynamics (\cite{heemelsaneel_lcss_2023}). This is needed because in the control practice only the states of the controller allow for projection, not the plant states. 

To define FHIGS, a projection operator is employed to keep the pair $(x_h,v_2)$ in the sector set $S$, or equivalently, to keep the states $x$ in the set
\small
\begin{equation}
    \cal S = \left\{ x \in \ree^{n+1} \mid x^\top\begin{bmatrix}
        1\\ 0_{n_{v_1}} \\-k_1C_{v_2}^\top \\k_1 C_p^\top D_{v_2}^\top
\end{bmatrix}\begin{bmatrix}
        1\\ 0_{n_{v_1}} \\-k_2C_{v_2}^\top \\k_2 C_p^\top D_{v_2}^\top
    \end{bmatrix}^\top x \leq 0   \right\},
\end{equation}
\normalsize
with $n := n_c+n_p$, $n_c:=n_{v_1}+n_{v_2}$. The projection operator is given by 
\begin{equation}\Pi_{\cal S,E}(\xi,p) := \argmin_{w \in T_{\cal S}(\xi), w - p \in \cal E} \|w-p\|. \label{piSE}
\end{equation}
Here, $T_{\cS}(\xi)$ is the tangent cone to the set $\cS \subset \mathbb{R}^{m+n}$ at a point $\xi \in \cS$,  defined as the collection of all vectors $p\in \mathbb{R}^{m+n}$ for which there exist sequences $\{x_i\}_{i\in \mathbb{N}} \in S$ and $\{\tau_i\}_{i\in \mathbb{N}}$, $\tau_i > 0$
with $x_i \rightarrow x$, $\tau_i \downarrow 0$ and $i \rightarrow \infty$, such that $p = \lim_{i\rightarrow \infty} \frac{x_i - x}{\tau_i}.$ The subspace $\cal E \subseteq \ree^{n+1}$ is the set of admissible directions for the projection, i.e., $w-p \in \cal E.$
The projected closed-loop dynamics are then written similarly to that in \cite{heemelsaneel_lcss_2023,DeeSha_AUT21a}
\begin{equation} \label{eq: cloop projected}
    \dot x = \Pi_{\cal S, \cal E} (x,f(x)),
\end{equation}
with $\cal E = \textrm{im}(E)$, $E = \begin{bmatrix}
    1 &0_n^\top
\end{bmatrix}^\top$, indicating that the projection operator only affects the $x_h$-dynamics (see \cite{heemelsaneel_lcss_2023} for more details and for the existence of solutions to \eqref{eq: cloop projected}).

\section{PWL dynamics of open-loop FHIGS}\label{sec: PWL openloop}
A piecewise linear (PWL) dynamics representation will be derived from the introduced ePDS representation. The PWL dynamics are used to explicitly compute the response to a sinusoidal input, leading to the calculation of the describing function of FHIGS in Section~\ref{sec: DF}. Also a PWL formulation of FHIGS allows for the application of analysis tools of hybrid systems, e.g., the use of piecewise quadratic Lyapunov functions and linear matrix inequalities to guarantee stability and performance of the closed-loop system (using adaptations of \cite{PWQLFforHybrid} along the lines of \cite{EijHee_NAHS22a}). For ease of calculations, we assume that $C_p B_p = 0$, i.e., $F_3$ and $G$ together have relative degree of $2$ or higher, which is reasonable in many settings, including motion control systems.

\begin{theorem}\label{theorem: pwl}
Consider the projected dynamics \eqref{eq: cloop projected} with \eqref{eq: open unprojected fhigs} and assume that $C_p B_p = 0$ and that $e$ is a differentiable function. Then, the FHIGS dynamics can be written as
\begin{equation}
    \dot x_c = A_i x_c + B_i \begin{bmatrix}e \\ \dot e \end{bmatrix}, \,\textrm{if}\, (x_c,e,\dot e) \in \cal F_i,\,i \in \{0,1,2\},
    % \dot x_c &= \,\textrm{if}\, (x_c,e,\dot e) \in \cal F_1 \\
    % \dot x_c &= \,\textrm{if}\, (x_c,e,\dot e) \in \cal F_2
\end{equation}
with
\small
\begin{align*}
    A_0 &= \begin{bmatrix}
        -\alpha_h &\omega_h C_{v_1} &0\\ 0 &A_{v_1} &0\\ 0&0 &A_{v_2}
    \end{bmatrix}, &B_0 = \begin{bmatrix}
        \omega_h D_{v_1} &0 \\ B_{v_1} &0\\
        B_{v_2} &0
    \end{bmatrix},\\
    A_1 &= \begin{bmatrix}
        0 &0 &k_1 C_{v_2} A_{v_2}\\ 0 &A_{v_1} &0\\ 0&0 &A_{v_2}
    \end{bmatrix}, &B_1 = \begin{bmatrix}
        k_1 C_{v_2} B_{v_2} &k_1 D_{v_2} \\ B_{v_1} &0\\
        B_{v_2} &0
    \end{bmatrix},\\
    A_2 &= \begin{bmatrix}
        0 &0 &k_2 C_{v_2} A_{v_2}\\ 0 &A_{v_1} &0\\ 0&0 &A_{v_2}
    \end{bmatrix}, &B_1 = \begin{bmatrix}
        k_2 C_{v_2} B_{v_2} &k_2 D_{v_2} \\ B_{v_1} &0\\
        B_{v_2} &0
    \end{bmatrix},
\end{align*}
\normalsize
and
\begin{align*}
    \cal F_1 = \{ (x_c,e,\dot e) &\in \ree^{n_c+3} \mid x_h = k_1 {v_2}\wedge \\ &v_2(k_1 \dot {v_2}- (-\alpha_h x_h + \omega_h e)) > 0 \},\\
    \cal F_2 = \{ (x_c,e,\dot e) &\in \ree^{n_c+3} \mid x_h = k_2 {v_2}\wedge \\ &v_2(k_2 \dot {v_2}- (-\alpha_h x_h + \omega_h e)) < 0 \},\\
    \cal F_0 = \{ (x_c,e,\dot e) &\in \ree^{n_c+3} \mid (x_h,v_2) \in S \wedge \\ &(x_c,e,\dot e) \notin (\cal F_1 \cup \cal F_2) \},
\end{align*}
with $v_2 = C_{v_2} x_{v_2}+ D_{v_2} e$, $\dot v_2 = C_{v_2} (A_{v_2} x_{v_2} + B_{v_2} e)+ D_{v_2} \dot e.$
\end{theorem}

The proof can be found in the appendix. 

Note the projected dynamics of the FHIGS are equivalent to a piecewise linear system with two gain modes (with gains $k_1$ and $k_2$) and a first-order-dynamics mode. The matrices $A_1,B_1,A_2,B_2$ are direct results of the two gain modes $\dot x_h = k_1 \dot v_2$ and $\dot x_h = k_2 \dot v_2$.

\section{Describing function analysis} \label{sec: DF}
In this section, the first-order describing function (DF) of FHIGS will be derived. To reduce the number of parameters, a simplified structure of FHIGS is introduced. The FHIGS structure is equivalent to the cascaded connection of the filter $F_1$ and the simplified structure (see Fig.~\ref{fig:simplifiedFHIGS}) if $F_1$ is an invertible filter and $F$ is selected as $F=F_1^{-1}F_2$. We also set $\alpha_h=0$ to obtain a closed-form solution of the DF of FHIGS. In addition to the parameters $k_1$, $k_2$, and $\omega_h$, FHIGS has additional parameters of the filter, namely, the magnitude and the phase of the transfer function $F(s)$ at frequency $\omega$, denoted $G(\omega) = |F(j\omega)|$, $\phi(\omega) = \angle F(j\omega)$, respectively.
\begin{figure}[tbp]
    \centering
    \smallskip\smallskip\includegraphics[width = 0.4 \textwidth]{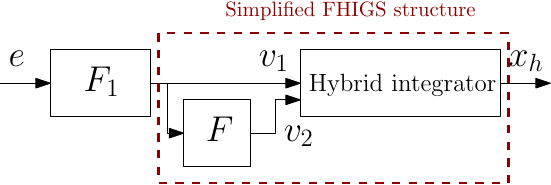}
    \caption{Simplified filtered HIGS structure with filter $F$. If $F=1$, FHIGS becomes the HIGS element.}
    \label{fig:simplifiedFHIGS}
\end{figure}

\begin{figure}[tbp]
    \centering
    \smallskip\smallskip\includegraphics[width = 0.45 \textwidth]{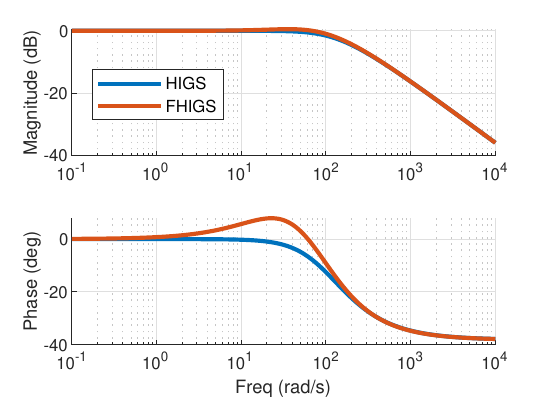}
    \caption{Comparison of the first-order describing function of HIGS and simplified FHIGS elements, with the parameters: $\omega_h = 100$, $k_h = 1$, $k_1 =0$, $k_2 = 1$, and FHIGS filter's transfer function $\frac{3(3s+2\omega_f)}{2(2s+3\omega_f)}$, with $\omega_f = 20 \pi$}
    \label{fig:HIGS compare}
\end{figure}

\subsection{Sinusoidal response with phase lead ($0\leq\phi(\omega)\leq \pi$)}\label{sec:response lead}
First, we consider a sinusoidal input $v_1 = A\sin{(\omega t)}$ and calculate the response of FHIGS in one period. Due to the two gain modes and integrator mode taking place sequentially (Fig.~\ref{fig: sinusoidal}), we find the steady-state response by simply integrating the dynamics in the integrator mode and taking $x_h = k_i v$ in the $k_i$-gain modes:
\small
\begin{equation}
\begin{aligned}
    &\frac{x_h(t)}{A} = \\&\begin{cases}
        k_1 G \sin{(\omega t + \phi)}, &\,\textrm{if}\, 0 \leq t \leq \frac{\epsilon}{\omega}\\
        k_1 G \sin{(\omega \epsilon + \phi)} + \frac{\omega_h(\cos{(\omega \epsilon)} - \cos{(\omega t)})}{\omega} &\, \textrm{if}\, \frac{\epsilon}{\omega}\leq t \leq \frac{\gamma}{\omega} \\
        k_2 G \sin{(\omega t + \phi)} &\, \textrm{if}\, \frac{\gamma}{\omega} \leq t \leq \frac{\pi - \phi}{\omega}\\
        k_1 G \sin{(\omega t + \phi)} &\,\textrm{if}\, \frac{\pi - \phi}{\omega} \leq t \leq \frac{\pi +\epsilon}{\omega}\\
        -k_1 G \sin{(\omega \epsilon + \phi)} - \frac{\omega_h(\cos{(\omega \epsilon)} - \cos{(\omega t)})}{\omega} &\, \textrm{if}\, \frac{\pi +\epsilon}{\omega}\leq t \leq \frac{\pi + \gamma}{\omega}\\
        k_2 G \sin{(\omega t + \phi)} &\, \textrm{if}\, \frac{\pi + \gamma}{\omega} \leq t \leq \frac{2\pi - \phi}{\omega}\\
        k_1 G \sin{(\omega t + \phi)} &\, \textrm{if}\, \frac{2\pi - \phi}{\omega} \leq t \leq \frac{2\pi}{\omega}
    \end{cases}
\end{aligned}
\end{equation}
\normalsize
\begin{figure}[tbp]
    \centering
    \includegraphics[width = 0.48\textwidth]{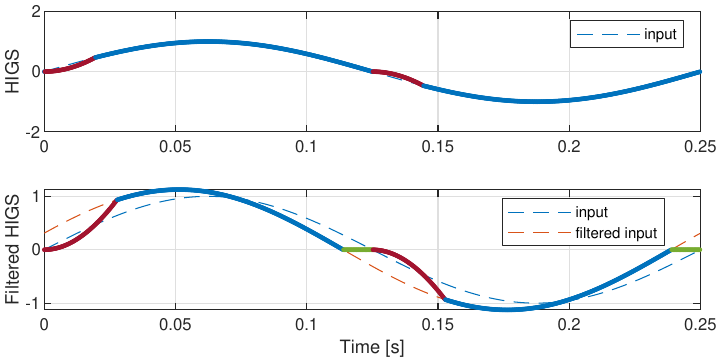}
    \caption{Steady-state responses of HIGS and FHIGS to a sinusoidal input (amplitude of 1, frequency of 4 Hz). The parameters of HIGS and FHIGS are $\omega_h = 100$, $k_h = 1$, $k_1 =0$, $k_2 = 1$, and the FHIGS filter has a transfer function of $\frac{3(3s+2\omega_f)}{2(2s+3\omega_f)}$, with $\omega_f = 20 \pi$. The multi-color lines show the responses of HIGS and FHIGS, and also indicate the active mode (red: integrator, green: $k_1$-gain, blue: $k_2$-gain (FHIGS) and gain (HIGS).}
    \label{fig: sinusoidal}
\end{figure}
The time instance $\epsilon(\omega)$ that determines the switching from the $k_1$-gain mode to the integrator mode (green to red in Fig.~\ref{fig: sinusoidal}) is the solution of the following equation
\begin{equation}
    \omega_h \sin{(\omega \epsilon) = k_1 G \omega \cos{(\omega \epsilon + \phi)}},
\end{equation}
which arises from the condition of the integrator mode being active on the $k_1$-line that $v(\omega_h e(\epsilon) - k_1 \dot v(\epsilon)) \geq 0$. The admissible solution that $0 \leq \epsilon \leq \pi/ \omega$ is
\begin{equation*}
    \epsilon = \frac{1}{\omega } {\cos ^{-1}\left(\frac{G k_1 \omega  \sin (\phi )+\omega _h}{\sqrt{G^2 k_1^2 \omega ^2+2 G k_1
   \omega  \omega _h \sin (\phi )+\omega _h^2}}\right)}
\end{equation*}
The other time instance $\gamma(\omega)$ that determines the switching from the integrator mode to the $k_2$-gain mode (red to blue in Fig.~\ref{fig: sinusoidal}) is found by solving the equality of the $k_2$-line:
\small
\begin{equation}\label{eq: solving gamma}
% \begin{aligned}
    k_2 G \sin{(\omega \gamma + \phi)} = k_1 G \sin{(\omega \epsilon + \phi)} + \frac{\omega_h(\cos{(\omega \epsilon)} - \cos{(\omega \gamma)})}{\omega}.
% \end{aligned}
\end{equation}
\normalsize

\subsection{Sinusoidal response with phase lag ($-\pi \leq \phi <0$)}\label{sec:response lag}
If the filter produces phase lag, the order of the modes changes. Two cases can happen: the $k_1$-gain mode is activated and not activated.
\subsubsection{The $k_1$ gain mode is activated}
This case happens when the $k_2$-gain mode can switch to the integrator mode during $t \in [0,-\phi/\omega]$. The response due to the sequence of modes (see Fig.~\ref{fig:sinusoidalLag}, top) is 
\small
\begin{equation}
\begin{aligned}
    &\frac{x_h(t)}{A} = \\&\begin{cases}
        k_2 G \sin{(\omega t + \phi)}, &\,\textrm{if}\, 0 \leq t \leq \frac{\epsilon}{\omega}\\
        k_2 G \sin{(\omega \epsilon + \phi)} + \frac{\omega_h(\cos{(\omega \epsilon)} - \cos{(\omega t)})}{\omega} &\, \textrm{if}\, \frac{\epsilon}{\omega}\leq t \leq \frac{\gamma}{\omega} \\
        k_1 G \sin{(\omega t + \phi)} &\, \textrm{if}\, \frac{\gamma}{\omega} \leq t \leq \frac{\pi - \phi}{\omega}\\
        k_2 G \sin{(\omega t + \phi)} &\,\textrm{if}\, \frac{\pi - \phi}{\omega} \leq t \leq \frac{\pi +\epsilon}{\omega}\\
        -k_2 G \sin{(\omega \epsilon + \phi)} - \frac{\omega_h(\cos{(\omega \epsilon)} - \cos{(\omega t)})}{\omega} &\, \textrm{if}\, \frac{\pi +\epsilon}{\omega}\leq t \leq \frac{\pi + \gamma}{\omega}\\
        k_1 G \sin{(\omega t + \phi)} &\, \textrm{if}\, \frac{\pi + \gamma}{\omega} \leq t \leq \frac{2\pi - \phi}{\omega}\\
        k_2 G \sin{(\omega t + \phi)} &\, \textrm{if}\, \frac{2\pi - \phi}{\omega} \leq t \leq \frac{2\pi}{\omega}
    \end{cases}
\end{aligned}
\end{equation}
\normalsize

\begin{figure}[tbp]
    \centering
    \smallskip\includegraphics[width=0.45\textwidth]{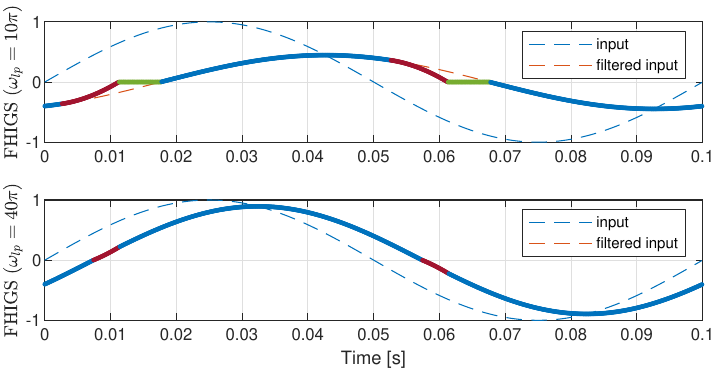}
    \caption{Steady-state responses of FHIGS to a sinusoidal input (amplitude of 1, frequency of 10 Hz). The parameters of HIGS and FHIGS are $\omega_h = 100$, $k_h = 1$, $k_1 =0$, $k_2 = 1$, and the FHIGS filter has a transfer function of $\frac{\omega_{lp}}{s+\omega_{lp}}.$}
    \label{fig:sinusoidalLag}
\end{figure}

The time instance $\epsilon(\omega)$ now determines the switching from the $k_2$-gain mode to the integrator mode, and is  solution to $\omega_h \sin{(\omega {\epsilon}) = k_2 G \omega \cos{(\omega {\epsilon} + \phi)}}.$
The other time instance $\gamma(\omega)$ now indicates the switching from the integrator mode to the $k_1$-gain mode, which is the solution to
\small
\begin{equation}\label{eq: gamma lag}
    k_1 G \sin{(\omega \gamma + \phi)} = k_2 G \sin{(\omega \epsilon + \phi)} + \frac{\omega_h(\cos{(\omega \epsilon)} - \cos{(\omega \gamma)})}{\omega}.
\end{equation}
\normalsize

\subsubsection{The $k_1$-gain mode is not activated}
This case happens when the $k_2$-gain mode cannot switch to the integrator mode during $t \in [0,-\phi/\omega]$. The response of FHIGS (see Fig~\ref{fig:sinusoidalLag}) becomes
\begin{equation}
    \frac{x_h(t)}{A} = \begin{cases}
        k_2 G \sin{(\omega t + \phi)}, &\,\textrm{if}\, 0 \leq t \leq \frac{-\phi}{\omega}\\
        \frac{\omega_h(\cos{(-\phi)} - \cos{(\omega t)})}{\omega} &\, \textrm{if}\, \frac{-\phi}{\omega}\leq t \leq \frac{\gamma}{\omega} \\
        k_2 G \sin{(\omega t + \phi)} &\,\textrm{if}\,  \frac{\gamma}{\omega}\leq t \leq \frac{\pi - \phi}{\omega}\\
         - \frac{\omega_h(\cos{(-\phi)} - \cos{(\omega t)})}{\omega} &\, \textrm{if}\, \frac{\pi - \phi}{\omega}\leq t \leq \frac{\pi + \gamma}{\omega}\\
        k_2 G \sin{(\omega t + \phi)} &\, \textrm{if}\, \frac{\pi + \gamma}{\omega} \leq t \leq \frac{2\pi}{\omega}
    \end{cases}
\end{equation}
\normalsize
with the switching instant $\gamma(\omega)$ being the solution to \eqref{eq: solving gamma} with $\epsilon = -\phi/\omega$.

\subsection{Describing function calculation} \label{sec:df calculation}
The $k$-th Fourier coefficients are calculated as
\begin{align*}
    a_k(\omega) &= \frac{\omega}{\pi} \int_0^{\frac{2\pi}{\omega}} x_h(t) \cos{(k\omega t)} dt,\\
    b_k(\omega) &= \frac{\omega}{\pi} \int_0^{\frac{2\pi}{\omega}} x_h(t) \sin{(k\omega t)} dt,
\end{align*}
and the $k$-th order describing function is then defined as \begin{equation}
    \cal D_k (\omega) = b_k(\omega) + j a_k(\omega).
\end{equation}
Similar to HIGS (\cite{EijkBeerHODFHIGS}), due to the response being an odd function, $\cal D_k(\omega) = 0$ for all $\omega >0$ if $k$ is even. The closed-form solution of $D_k$ for odd $k$ is given in the appendix.

\begin{remark}\label{remark: responses}
    Note that for sinusoidal inputs and $\alpha_h=0$, we computed here explicitly steady-state solutions for different cases. Hence, we proved the existence. In Section~\ref{sec: convergence} below we will also establish the uniqueness of this steady-state solution, thereby showing there are no other steady-state responses, which is important for the formalization of the describing function below. The uniqueness of the steady-state response will be proven by establishing incremental attractivity properties, which also proves then that the calculated steady-state response above is attractive in the sense that irrespective of the initial state of the FHIGS all solutions will eventually converge to this unique steady state.
\end{remark}

\subsection{Comparison with HIGS's DF}
A comparison of the describing function of the filtered HIGS and the original HIGS elements is given in Fig.~\ref{fig:HIGS compare}. The HIGS and FHIGS elements being compared are the ones that produce the response in Fig.~\ref{fig: sinusoidal}. With the added filter, FHIGS has a phase advantage over HIGS, and at some frequencies, the DF of FHIGS even has a positive phase lead. This phase lead is not possible for HIGS without an additional filter (e.g., a lead-lag filter) to process its input or output. However, the additional filter alters the magnitude at high frequencies since it obeys the Bode gain-phase relationship. FHIGS, on the other hand, does not change the magnitude at high frequencies.

% \begin{figure}[tbp]
%     \centering
%     \includegraphics[width = 0.5 \textwidth]{Figs/hodf_fhigs.eps}
%     \caption{Higher order describing functions of FHIGS with the same parameters as in Fig.~\ref{fig:HIGS compare}}
%     \label{fig:hodf}
% \end{figure}

\section{Incremental attractivity of open-loop FHIGS} \label{sec: convergence}
To prove the uniqueness of the steady-state solution computed in Section~\ref{sec: DF}, we show the incremental attractivity of FHIGS. For ease of calculation, we focus on the simplified structure of FHIGS (see Fig.~\ref{fig:simplifiedFHIGS}) with filter states $x_v \in \ree^n$, filter matrices $(A_v,B_v,C_v,D_v)$, and filter output $v_2 = C_v x_v + D_v v_1$. The states of FHIGS then become $x_c = \begin{bmatrix}
    x_h &x_v^\top
\end{bmatrix}^\top$. We state the incremental attractivity of open-loop FHIGS in the next theorem.

\begin{theorem}
    If $\alpha_h \geq 0$, $k_2> 0 \geq k_1$, and $A_{v}$ 
    % {\bf $A_v$ or $A_{v_1} - PLEASE CHECK} 
    is a Hurwitz matrix, given any different initial conditions $x_c(0)$ and $x'_c(0)$, the responses of FHIGS to the same bounded and differentiable input $v_1$, denoted $x_c(t,v_1,x_c(0))$ and $x_c(t,v_1,x'_c(0))$, converge, i.e.,
    \begin{equation*}
        \lim_{t\to\infty} \|x_c(t,v_1,x_c(0))-x_c(t,v_1,x'_c(0))\| = 0.
    \end{equation*}
\end{theorem}

\begin{proof}[Sketch of proof]
    We show that the differences $\xh - \xh'$ and $
        v_2 - v_2'
    $ converge to a box in finite time, and the box gets smaller as time increases and its size becomes zero at infinite time. Denote $\delta x_v = x_v - x'_v$, $\dv = v_2 - v_2'$, and $\dxh = \xh - \xh'$. First, $\| \delta x_v \|$ and $|\dv|$ have upper bounds that exponentially decrease over time since $A_{v}$ is Hurwitz. 
    Define two sets
    \small
\begin{align*}
    \Omega_1 &= \{ (\dxh, |\dv|) \in \ree^2 \mid (\dxh - k_1 |\dv|) (\dxh - k_2 |\dv|) \leq 0 \},\\
    \Omega_2 &= \ree^2 \setminus \Omega_1.
\end{align*}
\normalsize
Then $\dxh^2$ can be shown to be decreasing if $(\dxh, |\dv|) \in \Omega_2$ as
%. By definition of the set $\Omega_2$, $(\dxh,\dv)$ violates the $[k_1,k_2]$-sector, thus in that case, $(x_h,v)$ and $(x'_h,v')$ must be in the same cone (either $K$ or $-K$), and $v,v' \neq 0$. In addition, two trajectories cannot be in the same gain mode ($k_1$ or $k_2$), because if they were, the pair $(\dxh, \dv)$ would be in the set $\Omega_1$. Next, all the other possibilities are analyzed to show that when $(\dxh,\dv) \in \Omega_2$, 
$\dot \dxh \dxh \leq -\alpha_h \dxh^2 \leq 0$ when $(\dxh, |\dv|) \in \Omega_2$. Now we show incremental stability similarly to the proof of \cite[Theorem 1]{EijHee_CDC23a}.
% \todo[inline]{If $\alpha_h = 0$, $\delta x_h^2$ is decreasing except for the case that both trajectories are in integrator mode. But this case cannot happen forever, so there will still be a finite time that $(\delta x_h, \delta v)$ goes inside a smaller box}

Let $\rho_i = \sup_{t_i \leq \tau \leq t} |\dv(\tau)|$, and $M_i = \{ (\dxh,|\dv|) \in \ree^2 \mid |\dv| \leq \rho_i \wedge |\dxh| \leq k_h \rho_i \}$, with $k_h = \max (|k_1|,|k_2|)$, and $i \in \mathbb{N}$. First, consider $t_0 = 0$ and the box $M_0$.
If $(\dxh(0),|\dv(0)|) \notin M_0$, there is some $t > t_0$ such that $(\dxh(t),|\dv(t)|) \notin M_0$. Because $|\delta_v(t)| \leq \rho_0$, it must hold that $|\dxh| > k_h \rho_0 \geq k_h |\dv(t)|$, thus $(\dxh(t),|\dv(t)|) \in \Omega_2$. It follows that if $\alpha_h >0$
$$
    \dot \dxh \dxh \leq -\alpha_h \dxh^2 \implies |\dxh| \leq \mathrm{e}^{-0.5 \alpha_h t} |\dxh(0)|,
$$
and if $\alpha_h = 0$, $\dot \dxh \dxh < 0$ if the two trajectories are in different modes, and $\dot \dxh \dxh \leq  0$ if the two trajectories stay in the integrator mode. Since $v_1$ is bounded, the integrator mode is guaranteed to end, and the two trajectories must be in different modes, or $(\dxh,|\dv|)\in \Omega_1$ in finite time.
Therefore there must exist a finite time $t'_0$ such that $(\dxh(t'_0),|\dv(t'_0)|) \in M_0$. If $(\dxh(0),|\dv(0)|) \in M_0$, a simple selection is $t'_0 = t_0$.

As $|\dv|$ has a shrinking upper bound, there exists a finite time $t_1 > t'_0$ such that $\rho_1 < \rho_0$. Then consider the smaller box $M_1$, and with the same reasoning, there must exist a finite time $t'_1$ such that $(\dxh(t'_1),|\dv(t'_1)|) \in M_1$. Now for finite $i$, there must exist a finite time $t'_i$ such that $(\dxh(t'_i),|\dv(t'_i)|) \in M_i$. Let $i \to \infty$, one finds
\begin{equation*}
    0 \leq \lim_{t \to \infty} |\dxh(t)| \leq \lim_{i \to \infty} k_h \rho_i = 0,
\end{equation*}
therefore $\lim_{t \to \infty} |\dxh(t)| = 0$. Combining with the exponential decay of $\| \delta x_v\|$ gives the final result.
\end{proof}

\begin{remark}
    This incremental stability result and the computed steady-state response to a sinusoidal input in Section \ref{sec: DF} proves that the steady-state response given a sinusoidal input exists, is unique, and asymptotically attractive, as already discussed in Remark~\ref{remark: responses}.
\end{remark}
% \todo[inline]{Maurice remarks on stronger properties $\alpha_h >0$} {\bf SEE suggestion below}

\begin{remark}
In case of $\alpha_h=0$, we can also prove incremental Lyapunov stability and in case of $\alpha_h>0$, we can establish incremental input-to-state stability and incremental global exponential stability properties \cite{Angeli2002}, by adapting the proofs of \cite{EijHee_CDC23a}. 
% {\bf [RIGHT REF?]}. 
We leave the proofs for future work. These properties, in case of $\alpha_h>0$ can be used for the existence of  unique bounded steady-state responses to bounded inputs (including sinusoidal inputs), which can be proven to be globally exponentially stable. In fact, for periodic inputs, the unique steady-state response will be periodic as well and will have the same period. Explicit analytical computation of the steady-state response as done in Sections~\ref{sec:response lead} and ~\ref{sec:response lag} for $\alpha_h=0$ is not possible for $\alpha_h>0$, and one would have to resort to numerical methods for computing them and the DF. We focussed here on the case of $\alpha_h=0$ because of the  steady-state solution that we could compute analytically.
\end{remark}

\section{Numerical examples} \label{sec: examples}
\subsection{Gain-loss mitigation}
The added filter in HIGS regulates the mode-switching behavior of FHIGS, thus mitigating the gain-loss problem caused by too frequent switching. HIGS is susceptible to this problem (\cite{HeeEij_CCTA21a}), as any nonlinear integrator, especially when high-frequency signals are mixed into the input. For example, consider an input
\begin{equation*}
    e(t) = \sin{(\omega_1 t)} + \sin{(\omega_2 t)},
\end{equation*}
with $\omega_1 = 2\pi$, $\omega_2 = 20\pi$. The low-frequency signal $\sin{(\omega_1 t)}$ is considered the \textit{pure} signal, while the $ \sin{\omega_2 t}$ is considered as a high-frequency noise. HIGS's response to $e(t)$ is nearly identical to zero (see Fig.~\ref{fig:gainloss}, top). This problem is referred to as gain loss. Adding lifting and de-lifting filters as pre- and post-filters is a strategy to overcome the gain-loss problem of HIGS (\cite{EijHee_CONF21a}). The lifting filter is chosen as an inverted notch filter with the transfer function
\begin{equation} \label{eq: notch filter}
    N(s) = \frac{\frac{s^2}{\omega_N^2} + \frac{2\beta_1 s}{\omega_N} + 1}{\frac{s^2}{\omega_N^2} + \frac{2\beta_2 s}{\omega_N} + 1},
\end{equation}
with $\omega_N = 2\pi$, $\beta_1 = 0.2$, $\beta_2 = 0.02$, and the de-lifting filter is the inverse of it ($N^{-1}(s)$). The notch filter amplifies the low-frequency part by 10 times, hence HIGS primarily switches based on this low-frequency content and experiences less gain-loss. However, due to the de-lifting filter, the output is significantly changed compared with the HIGS's response to the pure signal. The filter of FHIGS is selected as a notch filter \eqref{eq: notch filter} ($\omega_N = 20\pi$, $\beta_1 = 0.02$, $\beta_2 = 0.2$) to attenuate the high-frequency signal. With this filter, FHIGS suffers less from gain-loss and provides an output comparable to the HIGS's response to the pure signal.
\begin{figure}[tbp]
    \centering
    \smallskip\smallskip\includegraphics[width=0.45\textwidth]{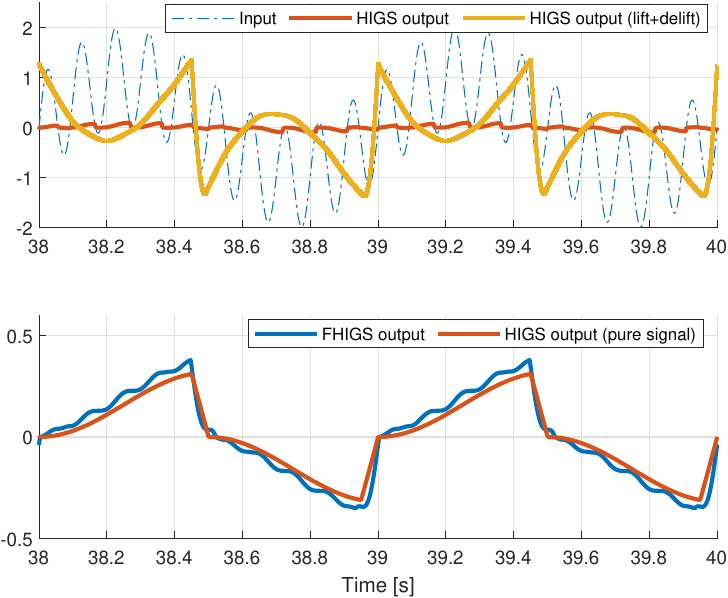}
    \caption{Gain loss mitigation. HIGS produces near zero output due to excessive switching (top, red). Gain-loss mitigation with lifting and de-lifting filters significantly alters the output (top, yellow), while FHIGS efficaciously reduces gain-loss with some mild changes in the output (bottom).}
    \label{fig:gainloss}
\end{figure}

\subsection{Control of a plant with a real pole to track a step input}
While HIGS has been shown to provide performance improvements to control systems, FHIGS can enhance the obtained benefits due to the added filtering. In this section, we provide an example of controlling a plant with a real pole to track a step reference. All linear controllers cause this plant to overshoot, while a HIGS-based controller may not \cite{SeronLimit,EijHee_CSL20a}. The transfer function of the plant is
\begin{equation}
    G(s) = \frac{1}{20s^2 - 5000}.
\end{equation}
A linear controller, a HIGS-based controller, and a FHIGS-based controller are tuned to control this plant. The HIGS-based controller has a HIGS-based lowpass filter (\cite{EijCCTA}) and a HIGS-based PI controller (\cite{EijHee_CSL20a}). The FHIGS-based controller replaces the HIGS element in the HIGS-based PI part with a FHIGS element. %The controller parameters, as well as the simulation files, are provided in our GitHub repository.
% \todo[inline]{write controller parameters}
In our example, FHIGS is shown to overcome the overshoot limitation of linear controllers and to perform slightly better in terms of settling time as compared to HIGS. The step responses of the plant are shown in Fig.~\ref{fig:simul}. The linear controller causes the plant to overshoot, while both HIGS and FHIGS do not. Moreover, they show a significant improvement in settling times with respect to the linear response.

\begin{figure}[htbp]
    \centering
    \includegraphics[width = 0.45 \textwidth]{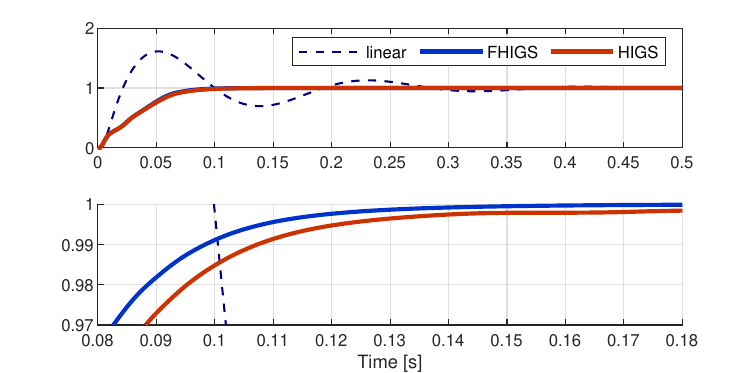}
    \caption{Step responses from 0.0s to 0.5s (top) and from 0.08s to 0.8s (bottom) with three controllers: linear, HIGS-based, and FHIGS-based}
    \label{fig:simul}
\end{figure}

\section{Conclusion} \label{sec: conclusion}
In this paper, we introduced a controller structure termed filtered hybrid integrator-gain system (FHIGS). This controller generalizes HIGS by adding linear filters in the nonlinear structure, allowing for more flexibility in controller tuning. FHIGS uses a projection operator to keep the filtered-input-output pair in a sector set, resulting in a three-mode operation. FHIGS was formally presented in the framework of extended projected dynamical systems and we derived a piecewise linear model descripton of it, which is useful for stability analysis and describing function computations. In fact, a quasi-linear analysis in terms of describing function is provided for FHIGS, which was supported by a formal analysis of incremental stability properties.
% {\bf see if you want to keep it, per discussion in chat}. 
Two numerical case studies show that FHIGS can even further enhance the performance benefits of HIGS. In particular, FHIGS can largely mitigate the gain-loss problem of HIGS with a single filter, instead of using a more complicated pre- and post-filtering scheme, and increased performance advantages are shown both in time- and frequency-domain. 

\bibliographystyle{IEEEtran}
\bibliography{bibBardia,bibOther,bibSebastiaan}

\appendix
\renewcommand{\baselinestretch}{}

\subsection{Proof of Theorem 1}
\begin{proof}
The sector set is represented as the union of two cones $\cal K := \{ (x_c,e) \in \ree^{n_{c}+1} \mid M \begin{bmatrix}
    x_c^\top & e
\end{bmatrix}^\top \geq 0 \}$ and $-\cal K := \{ (x_c,e) \in \ree^{n_{c}+1} \mid -M \begin{bmatrix}
    x_c^\top & e
\end{bmatrix}^\top \geq 0 \}$, with
\begin{equation}
    M = \begin{bmatrix}
        1 & 0_{n_{v_1}}&-k_1 C_{v_2}&-k_1 D_{v_2}\\ -1& 0_{n_{v_1}} &k_2 C_{v_2}& k_2 D_{v_2}
    \end{bmatrix}.
\end{equation}
We also use the index number $I(x_c,e) = \{ i \in \{1,2\} \mid F_i \begin{bmatrix}
    x_c^\top & e
\end{bmatrix}^\top = 0 \}$, with $F_i$ being the $i$-th row of $F$, $i \in \{1,2\}$.

Next, we explicitly compute the $x_c$-dynamics using the KKT optimality condition as in \cite{DeeSha_AUT21a} to obtain
\begin{equation}
\begin{bmatrix}
    \dot x_c \\ \dot e
\end{bmatrix}= \begin{bmatrix}
    f_c(x_c,e) +
     \bar{E}\eta^\star \\ \dot e
\end{bmatrix} =: f'_c(x_c,e,\dot e) +
     \bar{E}'\eta^\star,
\end{equation}
with $\bar{E} = \begin{bmatrix}
    1 &0_{n_{c}}^\top
\end{bmatrix}^\top$, $E' = \begin{bmatrix}
    \bar{E}^\top &0
\end{bmatrix}^\top,$
\begin{equation}
    \eta^\star = (E'^\top E')^{-1} E'^\top F_{I(x_c,e)}^\top \lambda_{I(x_c,e)}.
\end{equation}
Next, we find the complementarity condition 
\begin{equation}
    0 \leq \lambda_{I(x_c,e)} \perp F_{I(x_c,e)} \left(f'_c(x_c,e,\dot e) + E\eta^\star\right) \geq 0.
\end{equation}
 If $F_{I(x_c,e)} f'_c(x_c,e,\dot e) \geq 0$, the solution $\lambda_{I(x_c,e)} = 0$ is admitted, and the projection is not active. If $F_{I(x_c,e)} f'_c(x_c,e,\dot e) < 0$, the solution is
\begin{equation}
    \lambda_{I(x_c,e)} = - F_{I(x_c,e)} f'_c(x_c,e,\dot e).
\end{equation}
Straightforward further calculations lead to
\small
\begin{itemize}
    \item $I(x_c,e) = 1$: $(x_c,e,\dot e) \in \cal F_1$
    \begin{equation*}
    \begin{aligned}
        \dot x_h &= w_h e - (w_h e - k_1( C_{v_2}(A_{v_2}x_{v_2}+ B_{v_2}e) + D_{v_2}\dot e))\\ &= k_1 C_{v_2}(A_{v_2}x_{v_2}+ B_{v_2}e) + k_1 D_{v_2}\dot e= k_1 \dot v_2
    \end{aligned}
    \end{equation*}
    \item $I(x_c,e) = 2$: $(x_c,e,\dot e) \in \cal F_2$
    \begin{equation*}
        \begin{aligned}
            \dot x_h &= w_h e + (-w_h e + k_2 (C_{v_2}(A_{v_2}x_{v_2}+ B_{v_2}e)+D_{v_2}\dot e))\\ &= k_2 C_{v_2}(A_{v_2}x_{v_2}+ B_{v_2}e) + k_2 D_{v_2}\dot e = k_2 \dot v_2
        \end{aligned}
    \end{equation*}
\end{itemize}
\normalsize
We obtain the matrices $A_1,B_2,A_2,B_2$ as stated in Theorem~\ref{theorem: pwl}
\end{proof}

\subsection{Switching instant $\gamma$ when $\phi\geq 0$}
The solution to \eqref{eq: solving gamma} that satisfies $0\leq \gamma \leq \pi/ \omega$ is 
\begin{equation*}
    \gamma= \frac{1}{\omega} \cos ^{-1}\left(\frac{\sqrt{H_+(\omega)}+K_+(\omega)}{G^2 k_2^2 \omega ^2+2 G k_2 \omega  \omega _h \sin (\phi )+\omega _h^2}\right)
\end{equation*}
with
\small
\begin{equation*}
\begin{aligned}
    H_+(\omega) &= G^2 k_2^2 \omega ^2 \cos ^2(\phi ) (G^2 k_2^2 \omega ^2-G^2 k_1^2
   \omega ^2 \sin ^2(\epsilon +\phi ) \\&-2 G k_1 \omega  \omega _h \cos (\epsilon ) \sin (\epsilon +\phi )+2 G k_2 \omega 
   \omega _h \sin (\phi )\\&+\omega _h^2 \sin ^2(\epsilon )),\\
   K_+(\omega) &= G k_1 \omega  \sin (\epsilon+\phi) \left(G k_2
   \omega  \sin (\phi )+\omega _h\right)\\&+G k_2 \omega  \omega _h \cos (\epsilon ) \sin (\phi )+\omega _h^2 \cos
   (\epsilon )
\end{aligned}
\end{equation*}

\subsection{Switching instant $\gamma$ when $\phi< 0$}
The solution to \eqref{eq: gamma lag} that satisfies $0\leq \gamma \leq \pi/ \omega$ is 
\begin{equation*}
    \gamma= \frac{1}{\omega} \cos ^{-1}\left(\frac{\sqrt{H_-(\omega)}+K_-(\omega)}{G^2 k_{1}^2 \omega ^2+2 G k_{1} \omega  \omega _h \sin (\phi )+\omega _h^2}\right)
\end{equation*}
with
\small
\begin{equation*}
\begin{aligned}
    H_-(\omega) &= G^2 k_{1}^2 \omega ^2 \cos ^2(\phi ) (G^2 k_{1}^2 \omega ^2-G^2 k_2^2
   \omega ^2 \sin ^2(\epsilon +\phi ) \\&-2 G k_2 \omega  \omega _h \cos (\epsilon ) \sin (\epsilon +\phi )+2 G k_{1} \omega 
   \omega _h \sin (\phi )\\&+\omega _h^2 \sin ^2(\epsilon )),\\
   K_-(\omega) &= G k_2 \omega  \sin (\epsilon+\phi) \left(G k_{1}
   \omega  \sin (\phi )+\omega _h\right)\\&+G k_{1} \omega  \omega _h \cos (\epsilon ) \sin (\phi )+\omega _h^2 \cos
   (\epsilon )
\end{aligned}
\end{equation*}

\subsection{Fourier coefficients $a_1$ and $b_1$}
% \begin{equation*}
%     \begin{split}
%     a_1(\omega) = \frac{1}{2 \pi  \omega } (G k_1 \omega  (2 \cos (\omega \gamma -\omega \epsilon -\phi )-2 \cos (\omega \gamma +\omega \epsilon +\phi )+2 \omega \epsilon  \sin (\phi )+\cos (2
%    \omega \epsilon +\phi )+2 \phi  \sin (\phi )-\cos (\phi ))+G k_2 \omega  (2 (-\omega \gamma -\phi +\pi ) \sin (\phi )+\cos (2
%    \omega \gamma +\phi )-\cos (\phi ))-2 \omega _h (\cos (\omega \epsilon ) (\sin (\omega \epsilon )-2 \sin (\omega \gamma ))+\omega \gamma +\sin (\omega \gamma
%    ) \cos (\omega \gamma )-\omega \epsilon )),
% \end{split}
% \end{equation*}
% and
% \begin{equation*}
%     b_1(\omega) = \frac{G k_1 \omega  ((\cos (\omega \epsilon )-2 \cos (\omega \gamma )) \sin (\omega \epsilon +\phi )+(\omega \epsilon +\phi ) \cos (\phi ))+G k_2
%    \omega  ((-\omega \gamma -\phi +\pi ) \cos (\phi )+\cos (\omega \gamma ) \sin (\omega \gamma +\phi ))+\omega _h (\cos (\omega \gamma )-\cos
%    (\omega \epsilon ))^2}{\pi  \omega }.
% \end{equation*}

\begin{equation*}
    a_1(\omega) = \frac{1}{2 \pi  \omega } A(\omega),
\end{equation*}
with
\begin{equation*}
    \begin{aligned}
        A(\omega) &= G k_1 \omega  (2 \textrm{c} (\omega \gamma -\omega \epsilon -\phi )-2 \textrm{c} (\omega \gamma +\omega \epsilon +\phi )\\&+2 \omega \epsilon  \textrm{s} (\phi )+\textrm{c} (2
   \omega \epsilon +\phi )+2 \phi  \textrm{s} (\phi )-\textrm{c} (\phi )) \\&+G k_2 \omega  (2 (-\omega \gamma -\phi +\pi ) \textrm{s} (\phi )+\textrm{c} (2
   \omega \gamma +\phi )\\&-\textrm{c} (\phi ))-2 \omega _h (\textrm{c} (\omega \epsilon ) (\textrm{s} (\omega \epsilon )-2 \textrm{s} (\omega \gamma ))+\omega \gamma \\&+\textrm{s} (\omega \gamma
   ) \textrm{c} (\omega \gamma )-\omega \epsilon )
    \end{aligned}
\end{equation*}
and
\begin{equation*}
    b_1(\omega) = \frac{1}{\pi  \omega } B(\omega),
\end{equation*}
with
\begin{equation*}
    \begin{aligned}
        B(\omega) &= G k_1 \omega  ((\textrm{c} (\omega \epsilon )-2 \textrm{c} (\omega \gamma )) \textrm{s} (\omega \epsilon +\phi )+(\omega \epsilon +\phi ) \textrm{c} (\phi )) 
        \\&+G k_2 \omega  ((-\omega \gamma -\phi +\pi ) \textrm{c} (\phi )+\textrm{c} (\omega \gamma ) \textrm{s} (\omega \gamma +\phi ))
        \\&+\omega _h (\textrm{c} (\omega \gamma )-\textrm{c}(\omega \epsilon ))^2
    \end{aligned}
\end{equation*}

% \todo[inline]{too long equation}
\end{document}